\documentclass[12pt]{article}

\usepackage{eurosym}
\usepackage{graphicx}
\usepackage{enumerate}
\usepackage{amsmath}
\usepackage{amsfonts}
\usepackage{amssymb}
\usepackage{amsfonts}
\usepackage{amssymb,amsmath,amsthm,epsfig,graphicx,natbib,subfigure,hyperref}
\usepackage{geometry}
\usepackage{caption}
\usepackage{color}
\usepackage{enumerate}
\usepackage{setspace}
\usepackage[textwidth=30mm]{todonotes}
\usepackage{sgame, tikz}
\usepackage{multirow,array}
\usepackage{rotating}
\usepackage{graphicx}
\usepackage{subfigure}
\usepackage{verbatim}
\usepackage{bbm}
\usepackage{natbib}
\usepackage[ruled,vlined]{algorithm2e}
\usepackage{algorithmic}

\definecolor{darkblue}{rgb}{0.0, 0.0, 0.55}
\hypersetup{colorlinks,linkcolor={darkblue},citecolor={darkblue},urlcolor={darkblue}} 
\allowdisplaybreaks

\bibpunct[: ]{(}{)}{;}{a}{}{,}
\newtheorem{theorem}{Theorem}
\newtheorem*{theorem*}{Theorem}
\newtheorem{lemma}{Lemma}

\newtheorem{corollary}{Corollary}

\newtheorem{definition}{Definition}
\newtheorem{example}{Example}
\newtheorem*{example*}{Example}

\newtheorem{proposition}{Proposition}

\SetKw{KwAnd}{and}

\usepackage{soul}
\usepackage[normalem]{ulem}

\newcommand{\mycite}[1]{\citeauthor{#1}, \citeyear{#1}}

\title{Public Goods Provision in Directed Networks:

A Kernel Approach}
\author{Jingmin Huang\thanks{School of Economics, Renmin University of China, China. \textit{\ Email}: \href{mailto:jingmin.huang@ruc.edu.cn}{jingmin.huang@ruc.edu.cn}} \and Yang Sun\thanks{Department of Economics, Southwestern University of Finance and Economics, China. \textit{Email}: \href{mailto:sunyang789987@gmail.com}{sunyang789987@gmail.com}} \and Fanqi Xu\thanks{Department of Finance, Southwestern University of Finance and Economics, China.  \textit{Email}: \href{mailto:xufanqi1011@gmail.com}{xufanqi1011@gmail.com}} \and Wei Zhao\thanks{School of Economics and Management, Tsinghua University, China. \textit{Email}: \href{mailto:wei.zhao@outlook.fr}{wei.zhao@outlook.fr}}}
\date{Dec. 2025}

\begin{document}

\maketitle

\begin{abstract}

This paper investigates the decentralized provision of public goods in directed networks. We establish a correspondence between kernels in graph theory and specialized equilibria in which players either contribute a fixed threshold amount or free-ride entirely. Leveraging this relationship, we derive sufficient conditions for the existence and uniqueness of specialized equilibria in deterministic networks and prove that specialized equilibria exist almost surely in large random networks. We further demonstrate that enhancing network reciprocity weakly expands the set of specialized equilibria without destroying existing ones. Moreover, we propose an iterative elimination algorithm that simplifies the network while preserving equilibrium properties. Finally, we show that a Nash equilibrium is stable only if it is specialized, thereby providing dynamic justification for our focus on this equilibrium class.

\end{abstract}

\newpage

\tableofcontents

\newpage

\section{Introduction}

The non-excludability and non-rivalry properties of public goods lead to market failure. To mitigate such distortions \citep{lindahl_just_1958}, understanding the incentives for decentralized public goods provision is essential. When externalities are local and flow along geographic or social networks, the seminal paper by \cite{bramoulle2007public} analyzes how symmetric network structure shapes contribution incentives. However, externality flows are inherently asymmetric in many settings. In information and knowledge diffusion, content on social media propagates unilaterally from influencers to followers. In supply chains, technical innovations typically spill over from upstream manufacturers to downstream distributors but not vice versa. Similarly, physical flows such as river pollution abatement exhibit strict directionality: upstream efforts benefit downstream agents, while the reverse does not hold.\footnote{A vast literature analyzes directed information flows, including \cite{Acemoglu2010}, \cite{CalvoArmengol2015}, and \cite{bloch2018rumors}. Empirically, \cite{Mas2009} highlights the asymmetry of peer effects by showing that “worker effort is positively related to the productivity of workers who see him.”} This paper therefore investigates how asymmetric externality structures—represented by directed networks—determine incentives for public goods provision.

This paper extends the model of \cite{bramoulle2007public} to directed networks. A fixed directed network determines the pattern of spillovers among players. Each player values a public good that is costly to produce, and benefits are non-excludable along directed edges: if player $j$ contributes, all players with outgoing edges to $j$ receive the benefit. Players simultaneously decide whether to contribute, anticipating these directional spillover effects. We focus on specialized equilibria where the population partitions into specialists who provide the public good and free-riders who contribute nothing.

Our central result establishes that the sets of specialists in specialized equilibria correspond precisely to kernels of the underlying digraph—independent sets such that every outsider has an outgoing edge to the set. This characterization, building on the kernel concept introduced by \cite{von2007theory} for majority games and matching theory, provides a unified framework for analyzing equilibrium existence and multiplicity. Leveraging established results in kernel theory, we show the existence and uniqueness of specialized equilibrium in acyclic digraphs and demonstrate that existence holds almost surely in large random networks.

The kernel framework yields several comparative-static and welfare insights. First, we demonstrate that greater network reciprocity—converting one-way links into bidirectional ones—weakly expands the set of specialized equilibria without destroying existing ones. This monotonicity contrasts with \cite{bramoulle2007public}, where denser networks may reduce equilibrium welfare by altering contribution incentives. In directed networks, symmetrization preserves all existing specialized equilibria and thus improves equilibrium welfare unambiguously. Moreover, we develop an iterative elimination algorithm that systematically removes players whose strategic choices are locally determined—those who must contribute due to network position, those who can safely free-ride, and those without strategic influence. This procedure reduces equilibrium analysis to residual networks and yields sharp existence and uniqueness results, particularly for acyclic and sparse structures.

Finally, we characterize equilibrium stability under best-response dynamics. We prove that only specialized equilibria can be stable, providing justification for our focus on this class. We further show that a specialized equilibrium is stable if and only if it is stable in the game under the residual network obtained from the iterative elimination algorithm. Moreover, the equilibrium is stable if it corresponds to a sufficiently strong kernel in the sense that each free-rider has multiple contributing neighbors.

    This work contributes to the extensive literature on network games, which can be categorized by strategic interaction: strategic complementarities \citep{Ballester2006}, strategic substitutes \citep{bramoulle2007public}, and settings where both coexist \citep{bourl2017}. For comprehensive surveys on network games, see \citet{Jackson2017} and \citet{Elliott2019}.
   
    Our paper is particularly related to network games with strategic substitution. For undirected networks, \cite{bramoulle2007public} provide the foundational framework for analyzing strategic substitutes and establish that maximal independent sets characterize specialized equilibria. \cite{Bramoulle2014} and \cite{Allouch2017} derive sufficient conditions on network structure to guarantee uniqueness of Nash equilibrium under linear and nonlinear best-response functions, respectively. \cite{Galeotti2010} extend strategic network games to endogenous network formation problems. More recently, \cite{Allouch2019} introduce capacity constraints on individual public good provision within the \cite{bramoulle2007public} framework, while \cite{Allouch2021} analyze the welfare effects of taxation in networked public goods.

    For directed networks with strategic substitution, several papers have made important contributions. \cite{lopez2013public} extend \cite{bramoulle2007public} to directed unweighted networks, showing that maximal independent sets are sufficient to determine specialized equilibria in the directed network. \cite{Elliott2019a} characterize efficient solutions in terms of network eigenvalues in strategic substitute games, where the network capturing marginal externalities is naturally directed and weighted. \cite{bayer2023best} study directed network games by introducing best-response potential functions and analyzing convergence to Nash equilibria under one-sided improvement dynamics. \cite{papadimitriou2023public} focus on the computational complexity of determining equilibrium existence and computing Nash equilibria in public goods provision games on directed networks.

    This paper contributes directly to the \cite{bramoulle2007public} framework by providing a complete characterization of specialized equilibria in directed networks. We employ the kernel concept---a natural generalization of maximal independent sets to directed graphs---to characterize specialized equilibria. Building on kernel theory, we analyze existence, uniqueness, multiplicity, and stability properties of specialized equilibria, deriving results that parallel and extend those of \cite{bramoulle2007public} to the directed setting.
   
    The rest of this paper is organized as follows. Section 2 develops the model framework. Section 3 characterizes specialized equilibria using kernels and analyzes their existence and multiplicity. Section 4 studies the stability of specialized equilibria under best-response dynamics. Section 5 concludes.

\section{Model}
Consider a set $N=\{1,2,\ldots,n\}$ of players embedded in a directed network (digraph) $\mathbf{G}=(g_{ij})_{i,j\in N}$. The element $g_{ij}\in \{0,1\}$ represents the relationship between players $i$ and $j$, where $g_{ij}=1$ indicates a directed edge from $i$ to $j$ and $g_{ij}=0$ otherwise. We assume no self-loops, i.e., $g_{ii}=0$ for all $i\in N$. 

Each player $i$ chooses a level of public good provision $e_{i}\geq 0$ at marginal cost $c>0$ and benefits from the contributions of their out-neighbors $N_{i}=\{j:g_{ij}=1\}$—those players to whom $i$ has directed edges. Given a strategy profile $\mathbf{e}=(e_{1},\ldots,e_{n})$, player $i$'s payoff is
\begin{equation}
U_{i}(\mathbf{e},\mathbf{G})=b\left(e_{i}+\sum_{j\in N_{i}}e_{j}\right)-ce_{i},
\end{equation}
where the benefit function $b(\cdot)$ satisfies $b(0)=0$, $b'(0)>c$, $b'>0$, and $b''<0$. We denote this network game by $\Gamma(\mathbf{G})$.\footnote{The setup is identical to that of \cite{bramoulle2007public} except that we allow directed networks, where $g_{ij}\neq g_{ji}$ is possible.} The existence of a Nash equilibrium is guaranteed by the compactness of strategy spaces and the continuity of payoff functions. Moreover, \cite{bramoulle2007public} show that when $\mathbf{G}$ is undirected, multiple types of Nash equilibria exist, among which specialized equilibria are of particular interest.

Let $e^{\ast}$ denote the threshold level of public good provision at which marginal benefit equals marginal cost, i.e., $b'(e^{\ast})=c$. We then have the following definition:

\begin{definition}[\mycite{bramoulle2007public}]
A strategy profile $\mathbf{e}=(e_{1},\ldots,e_{n})$ is a specialized equilibrium of $\Gamma(\mathbf{G})$ if and only if for each player $i$, 
\begin{equation}
e_{i}=
\begin{cases}
0, & \text{if } \sum_{j\in N}g_{ij}e_{j}\geq e^{\ast}, \\
e^{\ast}, & \text{otherwise}.
\end{cases}
\label{D1}
\end{equation}
\end{definition}

A specialized equilibrium is characterized by a binary contribution pattern: each player either provides the threshold amount $e^{\ast}$ (specialists) or contributes nothing (free-riders). There are no intermediate contribution levels between $0$ and $e^{\ast}$. When the network is undirected, specialized equilibria are the only equilibria stable under small perturbations \citep{bramoulle2007public} and emerge more naturally than other equilibrium types in experimental settings \citep{rosenkranz2012network}.

A maximal independent set $I\subseteq N$ is a set of mutually unconnected players that is not a proper subset of any other set of unconnected players. When $\mathbf{G}$ is undirected, \citet{bramoulle2007public} establish that each specialized equilibrium corresponds to a maximal independent set of the network. Consequently, specialized equilibria always exist in undirected networks since every such network admits at least one maximal independent set. However, for directed networks, specialized equilibria may fail to exist, as the following example illustrates.

\begin{example}\label{E1}
Consider a directed three-player cycle depicted in Figure \ref{F1}, where player $j$ benefits from $i$'s provision, $k$ benefits from $j$'s provision, and $i$ benefits from $k$'s provision. No specialized equilibrium exists because any player's contribution generates a deviation cycle. Specifically, if $i$ contributes, then $j$ free-rides (benefiting from $i$'s provision), which induces $k$ to contribute. Player $k$'s contribution then creates an incentive for $i$ to deviate to free-riding.\footnote{Although no specialized equilibrium exists in this game, a Nash equilibrium does exist since $\Gamma(\mathbf{G})$ has compact strategy spaces and continuous payoffs. In particular, the profile where each player provides $\frac{e^{\ast}}{2}$ constitutes the unique (non-specialized) Nash equilibrium.}
\end{example}

\begin{figure}[h]
    \centering
    \includegraphics[width=0.23\linewidth]{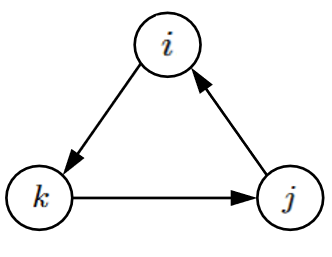}
    \caption{A directed cycle involving three players}
    \label{F1}
\end{figure}

\section{Characterization of Specialized Equilibria}

In this section, we characterize specialized equilibria in directed networks and analyze their existence.

\subsection{Specialized Equilibria and Kernels}
\begin{definition}[Kernel]\label{D2}
A set $K$ of nodes in a digraph $\textbf{G}$ is a kernel if and only if
\begin{enumerate}
    \item[(i)] $K$ is an independent set;
    \item[(ii)] for every $i\in N\setminus K$, there exists $j\in K$ such that $g_{ij}=1$.
\end{enumerate}
\end{definition}

A kernel requires that each node outside $K$ has at least one outgoing edge to a node within $K$. This concept naturally extends maximal independent sets from undirected to directed networks. In an undirected network, maximal independent sets and kernels coincide: a set is maximal independent if and only if it is independent and every outside node is adjacent to at least one inside node. The kernel generalizes this by imposing directionality: adjacency is replaced by the requirement that connections point \emph{toward} the kernel. Originally introduced by \cite{von2007theory} as a solution concept for majority games, kernels have since found applications in matching theory, logic, and other fields. The following theorem establishes the precise connection between kernels and specialized equilibria.

\begin{theorem}[Kernel and Specialized Equilibrium]\label{T1}
    A specialized profile is a Nash equilibrium if and only if its set of specialists forms a kernel of the directed network $\textbf{G}$. 
\end{theorem}

The economic intuition behind Theorem \ref{T1} is straightforward. The independence property (first point in Definition \ref{D2}) ensures that specialists within the kernel receive no contributions from other specialists and thus must provide the threshold level $e^{\ast}$ themselves. The domination property (second point in Definition \ref{D2}) guarantees that each non-specialist outside the kernel benefits from the contribution of at least one specialist in-neighbor, receiving sufficient externalities to optimally free-ride by contributing nothing. Unfortunately, \citet{chvatal1973computational} established that the kernel existence problem—determining whether a given digraph admits a kernel—is NP-complete. More recently, \citet{papadimitriou2023public} reiterated this result using a different proof technique in their Theorem 1.\footnote{Their analysis concerns pure-strategy equilibria in the indivisible-good setting with the max utility function. Although their model differs in formulation, players' best responses are identical to ours, so their result applies directly to specialized equilibria in our framework.}

\begin{corollary}[\mycite{papadimitriou2023public}]
    It is NP-complete to decide whether a public good game admits a specialized equilibrium.
\end{corollary}

There exists a relatively systematic body of work on the existence and uniqueness of kernels in directed networks. Below we present some results that provide conditions for existence and uniqueness.

\begin{proposition}[\mycite{richardson1946weakly}]\label{P1}
    A specialized equilibrium always exists in digraphs that contain no odd directed cycles.
\end{proposition}

Proposition \ref{P1} establishes a sufficient condition for the existence of a kernel in a digraph, and hence for the existence of a specialized equilibrium in the public goods game. Intuitively, the absence of odd directed cycles eliminates free-riding deviation cycles such as the one illustrated in Example \ref{E1}. An immediate corollary is that specialized equilibria exist in acyclic digraphs, including directed trees, transitive digraphs, and bipartite digraphs with exclusively directed links. 

Additional sufficient conditions, though more complex, are summarized in Theorem 2.3 of \cite{boros2006perfect}. Most of these relaxed conditions embody the intuition that when odd directed cycles contain sufficient additional connections, free-riding deviation cycles can be avoided, guaranteeing the existence of specialized equilibria. For instance, the presence of just two reversible arcs in each odd directed cycle suffices to ensure the existence of a kernel \citep{duchet1980graphes}. In Example \ref{E1}, adding a single reversible connection such as $g_{ij}=1$ yields a specialized equilibrium with $e_j=e^*$ and $e_i=e_k=0$. To date, no tractable necessary and sufficient condition for kernel existence has been established in the literature.\footnote{\cite{berge1973graphs} provided a necessary and sufficient condition for kernels in Proposition 1 (page 309), which characterizes kernels in terms of their characteristic functions, but this is merely a reformulation. Proposition 18 and Theorem 19 in the section ``Domination in Digraphs" of \cite{haynes2021structures} establish the equivalence between kernel-perfect digraphs and the existence of Grundy functions. However, characterizing the latter also lacks transparent necessary and sufficient conditions.} Using different proof techniques, several other sufficient conditions have been established. For example, Theorem 12.2 in \cite{fleiner2000stable} (or Theorem 2.4 of \cite{boros2006perfect}) provides conditions from the perspective of unions of partial orders, while Theorem 1 in \cite{boros1996perfect} approaches the problem through perfect graph theory and graph coloring.

Although the potential non-existence of kernels is somewhat discouraging, an interesting result indicates that kernels are highly likely to exist in random graphs when the number of players is sufficiently large. Let $\mathbf{G}(n,p)$ denote the random digraph with $n$ nodes where each directed edge is present independently with probability $p$.

\begin{proposition}[\mycite{de1990kernels}]\label{P2}
    For any fixed $p\in[0,1]$, the probability that a specialized equilibrium exists in the random directed graph $\mathbf{G}(n,p)$ tends to $1$ as $n\to \infty$.
\end{proposition}

Furthermore, the number $k(\mathbf{G})$ of specialized equilibria in digraph $\mathbf{G}$ is characterized by the following proposition.

\begin{proposition}[\mycite{boros2006perfect}]\label{P3}~~
    \begin{enumerate}
        \item[(i)] if all directed cycles in $\textbf{G}$ are odd, then $k(\textbf{G}) \leq 1$;
        \item[(ii)] if all directed cycles in $\textbf{G}$ are even, then $k(\textbf{G}) \geq 1$;
        \item[(iii)] if $\textbf{G}$ is acyclic, then $k(\textbf{G}) = 1$.
    \end{enumerate}
\end{proposition}

According to part (iii) of Proposition \ref{P3}, a unique specialized equilibrium exists when the network is an acyclic digraph. Note that in an acyclic digraph, there is no undirected links, as any bidirectional connection between nodes $i$ and $j$ (where $g_{ij}=g_{ji}=1$) would create a cycle $(i,j,i)$. In contrast, since any undirected graph admits a maximal independent set containing any given node, multiple specialized equilibria exist when the network is undirected.

\begin{example}\label{E2} 
Consider the simplest non-trivial undirected graph of two players with one connection, represented as a digraph, shown on the left of Figure \ref{F2}. It has a unique directed cycle $\{1,2,1\}$ which is even. There are two specialized equilibria $\mathbf{e}=(e^*,0)$ and $\mathbf{e}=(0,e^*)$. The set of contributors, $\{1\}$ or $\{2\}$, serves as a kernel in this graph. In the middle digraph of Figure \ref{F2}, no directed cycle exists and $\mathbf{e}=(e_1,e_2,e_3)=(0,0,e^*)$ is the only specialized equilibrium, corresponding to the unique kernel $\{3\}$. The right digraph of Figure \ref{F2} is the same as we have shown in Example \ref{E1}. It contains only one directed cycle, $\{1,3,2,1\}$, which has odd length. Consequently, no kernel exists, and no specialized equilibrium exists.
\end{example}

\begin{figure}[h]
    \centering
    \includegraphics[width=0.8\linewidth]{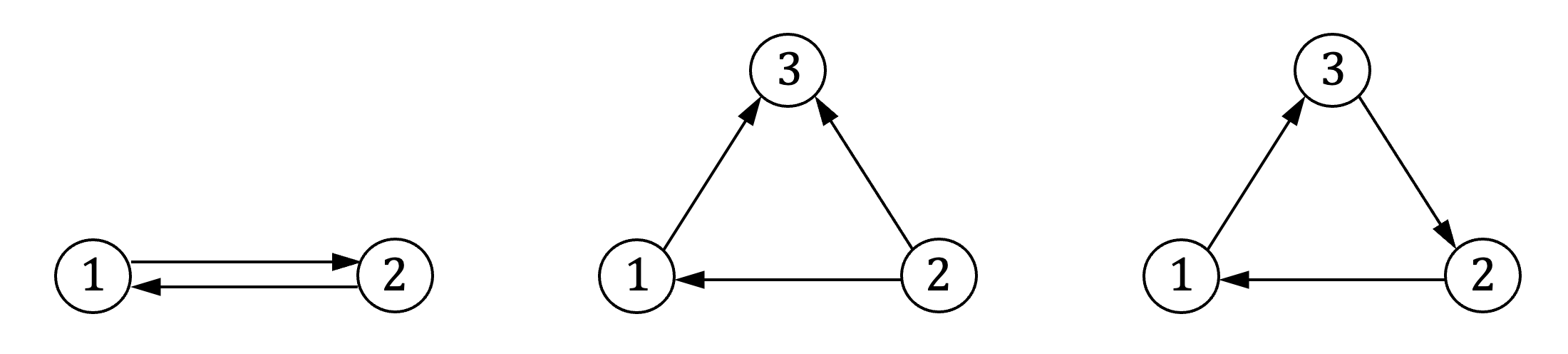}
    \caption{Cycles and Kernels}
    \label{F2}
\end{figure}

\subsection{Network Reciprocity}
To understand how the asymmetry of directed networks affects equilibrium outcomes, we introduce concepts that formalize the process of making directed edges bidirectional. For two networks $\mathbf{G}$ and $\mathbf{\hat{G}}$, we write $\mathbf{G}\subseteq \mathbf{\hat{G}}$ if $g_{ij}\leq \hat{g}_{ij}$ for all $i,j \in N$. 

\begin{definition}[Symmetrization and Partial Symmetrization]
The symmetrization of a digraph $\mathbf{G}$ is the undirected network $s(\mathbf{G})=\bar{\mathbf{G}}$, where $\bar{g}_{ij}=\max\{g_{ij},g_{ji}\}$ for each $i,j\in N$. A digraph $\mathbf{\hat{G}}$ is a partial symmetrization of $\mathbf{G}$, denoted $\mathbf{\hat{G}} \in ps(\mathbf{G})$, if $\mathbf{G}\subseteq \hat{\mathbf{G}}\subseteq s(\mathbf{G})$.
\end{definition}

The symmetrization $s(\mathbf{G})$ is the unique undirected graph obtained by making \textit{all} one-way edges bidirectional, while the set $ps(\mathbf{G})$ consists of all digraphs obtained by making \textit{some} one-way edges bidirectional. For two networks $\mathbf{G}$ and $\mathbf{G}^{\prime}$ with the same symmetrization $s(\mathbf{G})=s(\mathbf{G}^{\prime})$, we say they share \textit{the same architecture}—they differ only in edge directionality. All graphs in $ps(\mathbf{G})$ share the same architecture. For two networks $\mathbf{G}$ and $\mathbf{\hat{G}}$ with the same architecture, we say $\mathbf{G}$ is \textit{less reciprocal} than $\mathbf{\hat{G}}$ if $\mathbf{G}\subseteq \mathbf{\hat{G}}$. This partial order captures the intuition that $\mathbf{\hat{G}}$ contains all directed edges in $\mathbf{G}$ while converting some one-way edges to bidirectional connections, thereby becoming less asymmetric.

The following result establishes that specialized equilibria exhibit a striking monotonicity property: they persist as networks become more reciprocal.

\begin{proposition}[Equilibrium Monotonicity in Network Reciprocity]\label{P4}
For any digraph $\mathbf{G}$, if $\mathbf{e}$ is a specialized equilibrium of $\Gamma(\mathbf{G})$, then $\mathbf{e}$ is also a specialized equilibrium of $\Gamma(\hat{\mathbf{G}})$ for any $\hat{\mathbf{G}}\in ps(\mathbf{G})$.
\end{proposition}

Proposition \ref{P4} reveals that specialized equilibria are robust to increases in network reciprocity: any specialized equilibrium of $\Gamma(\mathbf{G})$ remains an equilibrium when directed edges are converted to bidirectional ones. The economic intuition is straightforward. Consider first a free-rider in the specialized equilibrium of $\Gamma(\mathbf{G})$. In the more reciprocal network $\hat{\mathbf{G}}$, this player receives a weakly larger amount of public good (since some contributors who previously could not reach them now can), making free-riding even more attractive as a best response. For a contributor in the specialized equilibrium of $\Gamma(\mathbf{G})$, the independence property ensures that this player has no outgoing connections with any other contributor. Converting edges bidirectional cannot create such connections, so providing $e^*$ remains optimal in $\hat{\mathbf{G}}$.

Importantly, this monotonicity is specific to specialized equilibria and fails for equilibria with interior contributions. Consider the distributed equilibrium $\mathbf{e}=(\frac{e^*}{2},\frac{e^*}{2},\frac{e^*}{2})$ in the directed cycle from Example \ref{E1}. This profile ceases to be an equilibrium in $\Gamma(s(\mathbf{G}))$ where all edges become bidirectional. In this fully reciprocal network, each player receives $e^*$ from their two neighbors, making it optimal to reduce their own contribution to zero rather than maintain $\frac{e^*}{2}$.

Proposition \ref{P4} immediately yields a sufficient condition for equilibrium existence. When networks $\mathbf{G}$ and $\mathbf{G}^{\prime}$ share the same architecture (i.e., $s(\mathbf{G})=s(\mathbf{G}^{\prime})$) with $\mathbf{G}^{\prime}\subseteq\mathbf{G}$, the game $\Gamma(\mathbf{G})$ admits a specialized equilibrium whenever $\Gamma(\mathbf{G}^{\prime})$ does. This nested relationship implies that within any architecture, \textit{strict} digraphs—those with only directed edges—represent the most restrictive case for equilibrium existence. If specialized equilibria exist in these least reciprocal networks, they necessarily persist as edges become bidirectional.

This observation allows us to characterize the relationship between specialized equilibria in undirected networks (corresponding to maximal independent sets) and those in strict digraphs (corresponding to kernels).

\begin{corollary} \label{C2}
    For any digraph $\mathbf{G}$, the set of contributors in each specialized equilibrium of $\Gamma(\mathbf{G})$ forms a maximal independent set of $s(\mathbf{G})$. Conversely, for any specialized equilibrium of the undirected game $\Gamma(\mathbf{\bar{G}})$, there exists a strict digraph $\mathbf{G}$ such that $s(\mathbf{G})=\mathbf{\bar{G}}$ and $\Gamma(\mathbf{G})$ admits the same specialized equilibrium.
\end{corollary}

The first statement follows directly from Proposition \ref{P4}: since any specialized equilibrium persists under symmetrization, the contributors must form a maximal independent set of $s(\mathbf{G})$. The second statement establishes the converse direction. Given any specialized equilibrium of the undirected game $\Gamma(\mathbf{\bar{G}})$ with contributor set $K$, we can construct a strict digraph $\mathbf{G}$ by orienting edges appropriately: convert bidirectional edges between contributors and free-riders into directed edges from free-riders to contributors, and arbitrarily orient edges between free-riders (which are payoff-irrelevant in this equilibrium). The resulting strict digraph admits the same specialized equilibrium.

These results provide a practical algorithm for verifying specialized equilibrium existence in $\Gamma(\mathbf{G})$: first identify all maximal independent sets of the symmetrization $s(\mathbf{G})$, then verify for each whether it can serve as a contributor set in a specialized equilibrium by checking the coverage condition.

\begin{example}\label{E3}
Consider the three digraphs $\mathbf{G}\subseteq \hat{\mathbf{G}}\subseteq s(\mathbf{G})$ illustrated in Figure \ref{F3}, which demonstrate how increasing network reciprocity affects the set of specialized equilibria. The strict digraph $\mathbf{G}$ contains only directed edges, the partially symmetric digraph $\hat{\mathbf{G}}$ adds some bidirectional connections, and $s(\mathbf{G})$ is the fully symmetric (undirected) network. 

The number of specialized equilibria increases monotonically with reciprocity: $k(\mathbf{G})=1$, $k(\hat{\mathbf{G}})=2$, and $k(s(\mathbf{G}))\ge 3$. The profile $\mathbf{e}=(e^*,0,0,0,0,e^*,0)$ is a specialized equilibrium common to all three networks, with players $1$ and $6$ as contributors (forming the kernel $\{1,6\}$). This equilibrium persists across all three networks, confirming Proposition \ref{P4}: specialized equilibria are preserved as reciprocity increases.

The partially symmetric network $\hat{\mathbf{G}}$ admits an additional specialized equilibrium $\mathbf{e}=(0,e^*,e^*,e^*,e^*,0,0)$ with contributor set $\{2,3,4,5\}$, which is also a specialized equilibrium of $s(\mathbf{G})$. The fully symmetric network $s(\mathbf{G})$ supports at least one more specialized equilibrium, such as $\mathbf{e}=(0,e^*,e^*,0,e^*,e^*,0)$ with contributor set $\{2,3,5,6\}$, which corresponds to another maximal independent set.

To illustrate the converse direction of Corollary \ref{C2}, consider the specialized equilibrium $(e^*,0,0,0,0,e^*,0)$ of $\Gamma(s(\mathbf{G}))$. The strict digraph $\mathbf{G}$ shown in the figure preserves this equilibrium, but it is not the unique strict digraph that does so. For instance, an alternative strict digraph with $g_{46}=1$ instead of $g_{64}=1$ (while maintaining all other edges) would also support the same specialized equilibrium, as long as each free-rider receives contributions from at least one specialist. This illustrates that the mapping from specialized equilibria of undirected games to supporting strict digraphs is generally one-to-many.

\begin{figure}[h]
    \centering
    \includegraphics[width=0.9\linewidth]{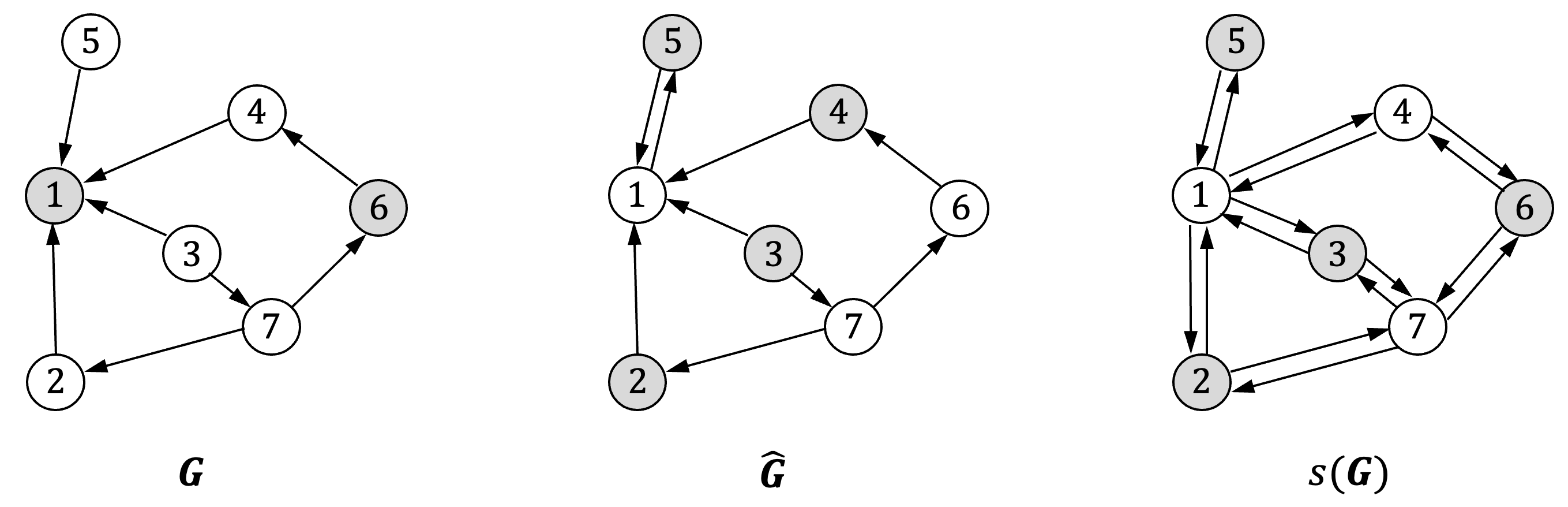}
    \caption{Network Reciprocity and Specialized Equilibria: Three networks with identical architecture but varying reciprocity levels.}
    \label{F3}
\end{figure}
\end{example}

\subsection{Elimination Algorithm}

We now present an alternative approach to identifying specialized equilibria that complements the kernel method. The key insight is that certain players' equilibrium actions can be determined independently of the full network structure. We exploit this by iteratively eliminating players whose strategic responses are unambiguous, reducing the problem to a simpler subgraph.

\begin{algorithm}[H]
\caption{Iterative Node Elimination}
\label{alg:node_elimination}
\begin{algorithmic}
\STATE Given a digraph $\mathbf{G}$, initialize $\mathbf{G}_0 = \mathbf{G}$ and set $t=0$.
\REPEAT
\STATE \textbf{Step 1:} Identify nodes without out-neighbors: $I = \{i \in N \mid \nexists j\in N \text{ s.t. } g_{ij}=1\}$.
\STATE \textbf{Step 2:} Identify nodes with at least one out-neighbor from $I$: $I^\prime = \{i' \in N \mid \exists i \in I \text{ s.t. } g_{i'i}=1\}$. 
\STATE \textbf{Step 3:} Eliminate nodes in $I \cup I'$ and their incident edges.
\STATE \textbf{Step 4:} Eliminate nodes without in-neighbors: $I'' = \{i \in N \mid \nexists j\in N \text{ s.t. } g_{ji}=1\}$.
\STATE Update $\mathbf{G}_{t+1}$ as the resulting subgraph and increment $t$.
\UNTIL{$\mathbf{G}_{t+1} = \mathbf{G}_t$}
\end{algorithmic}
\end{algorithm}

The algorithm terminates in finite time since $n$ is finite and each iteration either strictly reduces the number of nodes or leaves the graph unchanged. We denote by $\mathbf{G}_{\infty}$ the subgraph remaining when the algorithm converges. The economic logic behind each elimination step is as follows. Nodes in $I$ (Step 1) cannot receive contributions from others, so they must contribute $e^*$ in any specialized equilibrium. Nodes in $I'$ (Step 2) are guaranteed to receive contributions from at least one specialist (those in $I$), so they optimally free-ride. Nodes in $I''$ (Step 4) have no influence on others' payoffs and can be ignored when determining the contributor set. This systematic identification of players with determinate strategies leads to our next result.

\begin{proposition}[Elimination Algorithm and Specialized Equilibrium]\label{P5}
    For any digraph $\mathbf{G}$:
    
    (i) If $\mathbf{G}_{\infty}$ is non-empty, then $\Gamma(\mathbf{G})$ has the same number of specialized equilibria as $\Gamma(\mathbf{G}_{\infty})$.

    (ii) If $\mathbf{G}_{\infty}$ is empty, then $\Gamma(\mathbf{G})$ has a unique specialized equilibrium.
\end{proposition}

The algorithm is particularly effective for acyclic digraphs. A fundamental property of acyclic digraphs is that every such graph (and hence every induced subgraph) contains at least one node with in-degree zero and at least one node with out-degree zero (\cite{bang2008digraphs}, page 13, Proposition 1.4.2). Therefore, when $\mathbf{G}$ is acyclic, each iteration strictly reduces the number of nodes until $\mathbf{G}_{\infty}=\emptyset$, confirming the existence and uniqueness of specialized equilibria in acyclic networks.

For general digraphs, the algorithm may terminate with a non-empty $\mathbf{G}_{\infty}$, requiring further analysis of the remaining subgraph. The relative effectiveness of the elimination algorithm versus the kernel approach depends on the network structure. While the kernel method applies universally, the elimination algorithm can be more efficient when it significantly reduces the problem size. Interestingly, cycles may disappear during the elimination process even when they exist in the original graph, as nodes on the cycle may be removed in Step 2 when they have contributors among their in-neighbors.

\begin{example}\label{E4}
Consider the digraph $\mathbf{G}$ shown in Figure \ref{F4}, which contains the odd cycle $\{1,8,4,3,2,1\}$. This cycle prevents direct application of the kernel approach since odd cycles admit no independent sets covering all vertices. However, the elimination algorithm efficiently resolves this complexity.

In the first iteration, Step 1 identifies $I=\{5\}$ (light grey) as the unique node without out-neighbors, so player $5$ must contribute $e^*$ in any specialized equilibrium. Step 2 then identifies $I'=\{4,6\}$ (dark grey) as nodes guaranteed to receive contributions from player $5$, making free-riding their unique best response. Step 3 eliminates $\{4,5,6\}$ from the network. Step 4 removes player $3$, who after the previous eliminations has no remaining out-neighbors and thus exerts no strategic influence. This yields subgraph $\mathbf{G}_1$ containing only $\{1,2,7,8\}$.

Repeating the algorithm on $\mathbf{G}_1$, we find that player $7$ has no out-neighbors (so must contribute), player $8$ can free-ride on player $7$, and both can be eliminated. This leaves the residual subgraph $\mathbf{G}_\infty$ consisting of the bidirectional edge $1 \leftrightarrow 2$. The game $\Gamma(\mathbf{G}_\infty)$ admits exactly two specialized equilibria: either player $1$ contributes while player $2$ free-rides, or vice versa. By Proposition \ref{P3}, $\Gamma(\mathbf{G})$ inherits these two equilibria. Tracing back through the eliminated players, the complete equilibria are $\mathbf{e}=(e^*,0,0,0,e^*,0,e^*,0)$ and $\mathbf{e}'=(0,e^*,0,0,e^*,0,e^*,0)$.
\end{example}

\begin{figure}[h]
    \centering
    \includegraphics[width=0.9\linewidth]{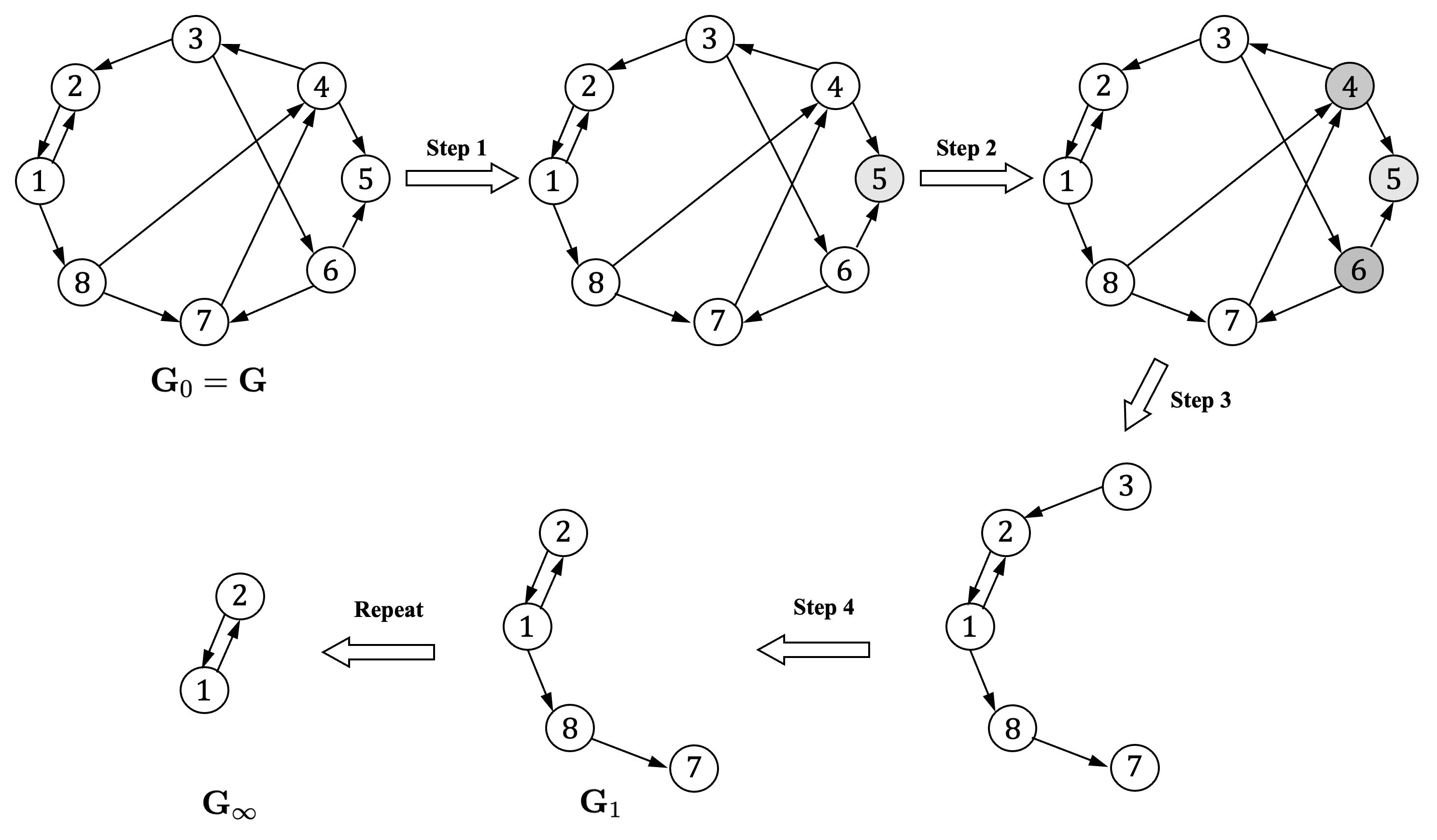}
    \caption{Elimination Algorithm and Specialized Equilibria}
    \label{F4}
\end{figure}

\section{Equilibrium Stability in Directed Networks}

Having analyzed the existence of specialized equilibria, we now examine which equilibria are stable under best-response dynamics. Following \cite{bramoulle2007public}, we adopt a natural refinement by excluding equilibria that are not robust to small perturbations.

Let $\mathbf{f}=(f_1,f_2,\dots,f_n)$ denote the best-response function profile, where $f_i(\mathbf{e})=\max\{0,e^*-\sum_{j=1}^n g_{ji}e_j\}$ for each player $i$. An equilibrium $\mathbf{e}$ is \textit{stable} if there exists $\rho>0$ such that for any bounded legitimate perturbation profile $\boldsymbol{\varepsilon}$ satisfying $|\varepsilon_i|\le \rho$ and $e_i+\varepsilon_i\ge 0$ for all $i$, the sequence $\{\mathbf{e}^{(r)}\}$ defined by $\mathbf{e}^{(0)}=\mathbf{e}+\boldsymbol{\varepsilon}$ and $\mathbf{e}^{(r+1)}=\mathbf{f}(\mathbf{e}^{(r)})$ converges to $\mathbf{e}$. Intuitively, a stable equilibrium attracts nearby strategy profiles under iterative best responses.

As in the undirected case, non-specialized equilibria are inherently unstable in directed networks, justifying our focus on specialized equilibria throughout this analysis.

\begin{lemma}\label{L1}
    Any non-specialized Nash equilibrium of $\Gamma(\mathbf{G})$ is unstable.
\end{lemma}

The instability of non-specialized equilibria stems from strategic substitution between contributors: when multiple players contribute positive but interior amounts, small perturbations trigger cascading adjustments that fail to converge back to the original equilibrium.

We exploit the elimination algorithm introduced in the previous section to analyze equilibrium stability. We use $\mathbf{e}^{\mathbf{G}}$ and $\mathbf{e}^{\mathbf{G}_\infty}$ to denote the strategy profiles in games $\Gamma(\mathbf{G})$ and $\Gamma(\mathbf{G}_{\infty})$, respectively, where $\mathbf{G}_\infty$ is the subgraph of $\mathbf{G}$ obtained from Algorithm \ref{alg:node_elimination} and $e_{i}^{\mathbf{G}_\infty}=e_{i}^{\mathbf{G}}$ for any player $i$ in subgraph $\mathbf{G}_\infty$. The following result states that stability is entirely determined by $\Gamma(\mathbf{G}_{\infty})$. We say a kernel $K$ has \textit{order} $k$ if every outsider has at least $k$ outgoing edges to nodes in the kernel, i.e., for all $i\in N\setminus K$, there exist at least $k$ distinct nodes $j_1,\dots,j_k\in K$ such that $g_{ij_{1}}=\dots=g_{ij_{k}}=1$.

\begin{proposition}\label{P6}
A specialized Nash equilibrium $\mathbf{e}^{\mathbf{G}}$ is stable if and only if $\mathbf{e}^{\mathbf{G}_\infty}$ is a stable equilibrium in $\Gamma(\mathbf{G}_{\infty})$. In particular, a specialized Nash equilibrium $\mathbf{e}^{\mathbf{G}}$ (corresponding to kernel $K$) is stable if either of the following conditions holds: 

    (i) $\mathbf{G}_\infty$ is empty;
    
   (ii) $K$ is a kernel of order $2$.
\end{proposition}

Since eliminated nodes have determinate responses, they cannot destabilize the equilibrium—stability is inherited from the remaining subgraph $\mathbf{G}_\infty$. When $\mathbf{G}_\infty$ is empty, all strategic choices are fully determined and independent of neighbors' choices, guaranteeing stability. 

Condition (ii) extends the result of \cite{bramoulle2007public} from undirected to directed networks. When each free-rider has at least two contributing neighbors, this redundancy in coverage provides robustness against perturbations: even if one contributor temporarily reduces their contribution, another ensures sufficient provision to prevent free-riders from switching strategies. The key difference from \cite{bramoulle2007public} is that in undirected networks, maximal independent sets of order 2 characterize stable equilibria, whereas in directed networks, kernels of order 2 provide only a sufficient but not necessary condition for stability, as illustrated by the following example.

\begin{example}\label{E5}
Panel (a) depicts the digraph studied in Example 4, where $\mathbf{G}_\infty$ is a clique formed by two players (see $\mathbf{G}_\infty$ in Figure \ref{F4}). Since there is no maximal independent set of order 2 in $\mathbf{G}_\infty$, there is no stable equilibrium in this game. 

Panel (b) presents a network where $\mathbf{G}_\infty=\mathbf{G}$ and the kernel $\{1,2,6,7\}$ has order 2, yielding a stable specialized equilibrium. In contrast, the set $\{3,4,5\}$ also forms a kernel but has only order 1, since node $6$, for instance, has only one outgoing edge $6\to 4$ to the kernel nodes. Consequently, the specialized equilibrium in which $\{3,4,5\}$ are contributors is not stable.

Panel (c) shows the simplest case where condition (ii) fails (the contributor set $\{1\}$ has order 1), yet the equilibrium is stable by condition (i) since $\mathbf{G}_\infty=\emptyset$. This example illustrates that, unlike \cite{bramoulle2007public} where maximal independent sets of order 2 characterize stable equilibria, kernels of order 2 provide only a sufficient condition for stability in directed networks.
\end{example}

\begin{figure}[h]
    \centering
    \includegraphics[width=0.9\linewidth]{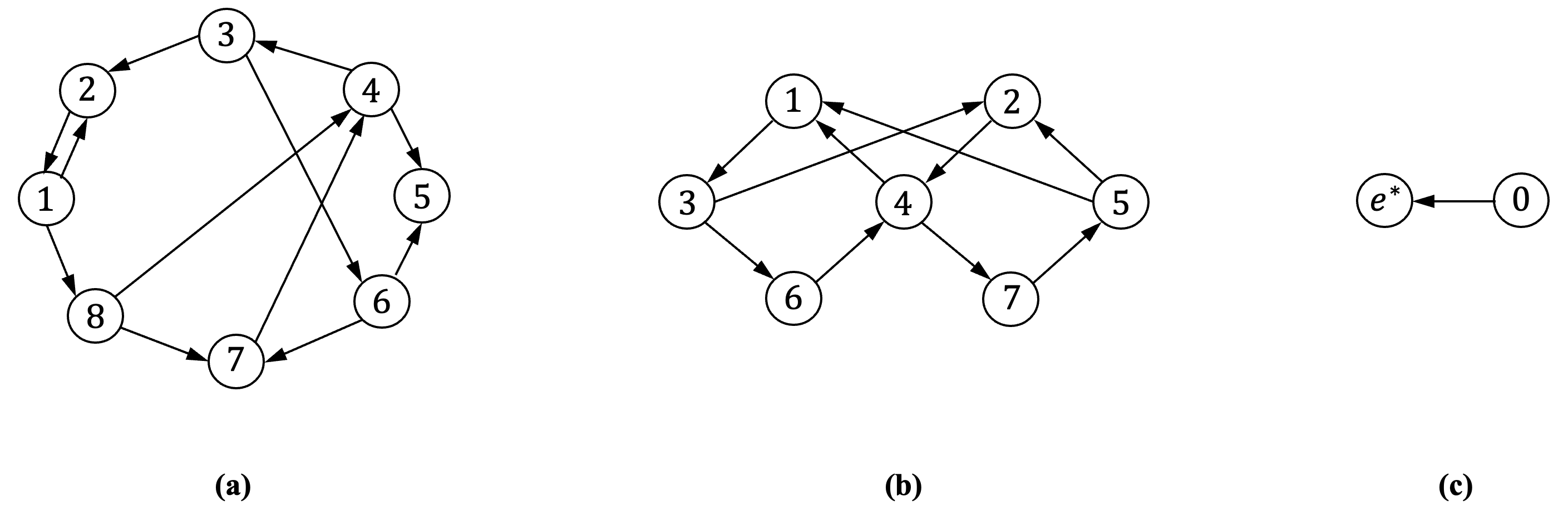}
    \caption{Specialized Equilibria and Stability.}
    \label{F5}
\end{figure}

\section{Discussion and Conclusion}

This paper extends the public goods provision framework of \cite{bramoulle2007public} from undirected to directed networks. We characterize specialized equilibria using the classical concept of kernels introduced by \cite{von2007theory}. Drawing on results from kernel theory in graph theory, we analyze the existence and multiplicity of specialized equilibria in relation to digraph structure. We establish that specialized equilibria exist in acyclic digraphs and almost surely in large random networks. We further show that specialized equilibria are robust to increases in network reciprocity: any specialized equilibrium remains an equilibrium when one-way edges are converted to bidirectional connections, though additional equilibria may emerge in the more reciprocal network. To facilitate equilibrium identification, we propose an iterative elimination algorithm that removes players whose best responses are determined, reducing the analysis to a simpler residual subgraph. Finally, we show that stability of specialized equilibria is preserved under this reduction, and that specialized equilibria are stable when kernels have order 2.

The kernel framework yields novel comparative-static insights. While \cite{bramoulle2007public} shows that adding edges in undirected networks can reduce equilibrium welfare by altering contribution incentives, our Proposition \ref{P4} establishes that converting one-way edges to bidirectional connections preserves all existing specialized equilibria. Consequently, when the welfare-maximizing equilibrium is specialized, bidirectional conversion necessarily improves equilibrium welfare—the preserved equilibrium generates strictly higher surplus due to expanded benefit flows, without creating deviation incentives for any player. This suggests that fostering reciprocity in directed networks with asymmetric relationships is unambiguously welfare-enhancing.

Our theoretical framework yields practical insights for the directed network settings motivating this study. In social media ecosystems, our results suggest that identifying key content creators (kernels) is crucial for understanding organic content provision, as these influencers naturally emerge as specialists while their followers optimally free-ride on generated content. The elimination algorithm provides a tractable method for identifying such influential nodes in large-scale networks. In supply chain innovation, Proposition \ref{P4} implies that fostering reciprocal knowledge-sharing relationships—converting one-way technical spillovers into bidirectional collaborations—weakly increases welfare without discouraging existing innovators. This welfare improvement occurs because reciprocity preserves existing innovation incentives while expanding the benefits from knowledge spillovers to a broader set of firms.

Our results extend naturally along several dimensions. First, for weighted networks, suppose $\mathbf{G}$ is a weighted digraph. A set $K$ of nodes is a kernel if and only if $K$ is an independent set and for all $i\in N\setminus K$, $\sum_{j\in K} g_{ij}\ge 1$. We can then establish a one-to-one correspondence between kernels of $\mathbf{G}$ and specialized equilibria in the public goods provision game under $\mathbf{G}$. The existence and uniqueness theory of kernels in weighted graphs has received limited attention in the literature and provides a novel direction for future research.

Second, our Proposition \ref{P2} demonstrates the existence of specialized equilibria as population size increases by utilizing existing results in kernel theory. We assume that each player knows the exact network structure and precisely with whom they interact, even as $n\to \infty$. Therefore, specialized equilibria are characterized by kernels in the large network. A classical model of large network games is the graphon game framework proposed by \cite{Parise2023}, in which players interact according to a graphon. It would be of interest to extend the concept of kernels to graphons and relate them to specialized equilibria in graphon games. We leave this as a direction for future research.

\newpage
\bibliographystyle{chicago}
\bibliography{ref.bib}

\newpage

\appendix
\renewcommand{\thesection}{\Alph{section}}

\section{Proofs in Section 3}

\subsection{The Proof of Theorem \ref{T1}}

Consider a specialized equilibrium where $K$ is the set of specialists. Specialists play a best-response if all their out-neighbors exert zero effort. For any pair of specialists $i,j\in K$, we must have $g_{ij}=0$; otherwise, specialist $i$ has at least an out-neighbor $j$ taking effort $e^*$. Thus, $K$ is an independent set of the graph. A non-specialist $j\in N\setminus K$ plays a best-response if $\sum_{i\in N_j}e_i\geq e^*$, in a specialized equilibrium which means agent $j$ is an in-neighbor of at least one specialist $i\in K$. Combining both properties yields the necessity. The sufficiency part can be checked directly.

\subsection{The Proof of Proposition \ref{P4}}

Fix a specialized equilibrium $\mathbf{e}$ of $\Gamma(\mathbf{G})$. On one hand, the specialists form an independent set, with $g_{ij}=g_{ji}=0$ for any pair of specialists $i$ and $j$. Thus, $\bar{g}_{ij}=\max\{g_{ij},g_{ji}\}=0$ in the undirected closure $s(\mathbf{G})$. Then, $\hat{g}_{ij}\le \bar{g}_{ij}=0$, which implies that $\hat{g}_{ij}=0$ for any pair of specialists. On the other hand, by the definition of specialized equilibrium, for any non-specialist $i$, $\sum_{j\in N}g_{ij}e_{j}\geq e^{\ast}$. Then, player $i$ in the partial closure $\hat{\mathbf{G}}$ faces total externalities $\sum_{j\in N}\hat{g}_{ij}e_{j}\geq\sum_{j\in N}g_{ij}e_{j}\geq e^{\ast}$. Combining both sides yields the result.

\subsection{The Proof of Corollary \ref{C2}}

Only the converse direction needs to be proved. Suppose $K$ is the set of contributors corresponding to a specific specialized equilibrium of $\Gamma(\mathbf{\bar{G}})$. Define a new digraph $\mathbf{G}$ by setting $g_{ij}=1$ if and only if in case $\bar{g}_{ij}=1$, $j\in K$ and $i\in N\setminus K$ or in case $\bar{g}_{ij}=1$, $i,j\in N\setminus K$ and $i>j$. $G$ is strict since $K\cap (N\setminus K)=\emptyset$. And $K$ is a kernel of the graph $\mathbf{G}$ guaranteed by that $K$ is a maximal independent set of the graph $\mathbf{\bar{G}}$. By Proposition \ref{P1}, the same specialized equilibrium corresponding to $K$ exists in the game $\Gamma(\mathbf{G})$.

\subsection{The Proof of Proposition \ref{P5}}

In a specialized equilibrium, players in set $I$ (with no outgoing edges) must be contributors, as they cannot benefit from any other players. Meanwhile, for a player $i'$ in set $I'$, suppose player $i\in I$ satisfies $i\in N_{i'}$. Since $i$ is a contributor in equilibrium, $i'$ must be a free-rider. Finally, players in $I''$ (with no incoming edges) do not affect the best responses of other players since they do not provide benefits to anyone. All in all, there is a one-to-one relationship between the specialized equilibrium of $\Gamma(\mathbf{G}_{t+1})$ and that of $\Gamma(\mathbf{G}_t)$. Proposition \ref{P5} then follows.

\section{Proofs in Section 4}

The proof of our Lemma \ref{L1} and Proposition \ref{P6} (ii) can be seen as a directed version of the corresponding proof in \cite{bramoulle2007public}. However, their proof that any specialized equilibrium $e$ in which exist an agent $i$ as a non-specialist connected to a unique specialist $j$ is always unstable does not work now. The reason is that their argument relies on a perturbation propagating along an edge and then propagating backward, which does not necessarily hold in the context of directed graphs. Consequently, we cannot obtain necessary and sufficient conditions for stable equilibria analogous to theirs.







Then, we prove the following lemma, and Proposition \ref{P6} follows.

\begin{lemma}\label{L2}
    $\mathbf{e}^{\mathbf{G}_{t+1}}$ is a stable specialized equilibrium of $\Gamma(\mathbf{G}_{t+1})$, if and only if $\mathbf{e}^{\mathbf{G}_{t}}$ is a stable specialized equilibrium of $\Gamma(\mathbf{G}_{t})$.
\end{lemma}
\begin{proof}
Only the following two simple cases need to be proved and Lemma \ref{L2} holds as a combination of them. To be clear, relabel the players in $\Gamma(\mathbf{G}_{t})$ as $M_t=\{1,2,\dots,m\}$.

\textbf{Case 1: Only nodes $I\cup I'$ is eliminated in step 3 in the $(t+1)$th iteration.} Suppose $M_t=M_{t+1}\cup I\cup I'$ where $M_{t+1}$ is the residual players in $\Gamma(\mathbf{G}_{t+1})$. Then,
\begin{equation}\label{eq3}
\begin{aligned}
    \mathbf{f}_t(\mathbf{e}^{\mathbf{G}_{t}}+ \boldsymbol{\varepsilon}) &= (\mathbf{f}_t^{M_{t+1}}(\mathbf{e}^{\mathbf{G}_{t}}+ \boldsymbol{\varepsilon}),\mathbf{f}_t^{I}(\mathbf{e}^{\mathbf{G}_{t}}+ \boldsymbol{\varepsilon}),\mathbf{f}_t^{I'}(\mathbf{e}^{\mathbf{G}_{t}}+ \boldsymbol{\varepsilon})) \\
    &= (\mathbf{f}_t^{M_{t+1}}(\mathbf{e}^{\mathbf{G}_{t}}+ \boldsymbol{\varepsilon}),\mathbf{e}^{*I},\mathbf{f}_t^{I'}(\mathbf{e}^{\mathbf{G}_{t}}+ \boldsymbol{\varepsilon}))
\end{aligned}
\end{equation}
where $\mathbf{f}_t^X: \mathbb{R}^{|M_t|}\to \mathbb{R}^{|X|}$ denotes the restriction of the response function $\mathbf{f}_t$ of the game $\Gamma(\mathbf{G}_t)$ on selected players in set $X\subseteq M_t$, and  $\mathbf{e}^{*X}$ and $\mathbf{0}^{X}$ means all players in $X$ take effort $e^*$ and $0$, respectively. The second equation of the formula (\ref{eq3}) holds by the definition of the set $I$. As the players in $I'$ must has a out-neighbor in $I$ and all players in $I$ take effort $e=e^*$ now,
\begin{equation*}
\begin{aligned}
    \mathbf{f}_t^{(2)}(\mathbf{e}^{\mathbf{G}_{t}}+ \boldsymbol{\varepsilon}) &=  \mathbf{f}_t(\mathbf{f}_t^{M_{t+1}}(\mathbf{e}^{\mathbf{G}_{t}}+ \boldsymbol{\varepsilon}),\mathbf{e}^{*I},\mathbf{f}_t^{I'}(\mathbf{e}^{\mathbf{G}_{t}}+ \boldsymbol{\varepsilon})) \\
    &= (\mathbf{f}_t^{M_{t+1}(2)}(\mathbf{e}^{\mathbf{G}_{t}}+ \boldsymbol{\varepsilon}),\mathbf{e}^{*I},\mathbf{0}^{I'})
\end{aligned}
\end{equation*}

On the one hand, if $\mathbf{e}^{\mathbf{G}_{t+1}}$ is stable, there exists $\rho>0$ such that $\mathbf{f}_{t+1}^{(n)}(\mathbf{e}^{\mathbf{G}_{t+1}}+ \boldsymbol{\hat{\varepsilon}})\to \mathbf{e}^{\mathbf{G}_{t+1}}$ if for any $i\in M_{t+1}$, $|\hat{\varepsilon}_i|\le \rho$ and $e_i+\hat{\varepsilon}_i\ge 0$. Let $\rho_1=\rho/m^2$, then for $\boldsymbol{\varepsilon}$ such that $|\varepsilon_i|\le \rho_1$ and $e_i+\varepsilon_i\ge 0$, $\exists \tilde{\boldsymbol{\varepsilon}}$ such that for any $i\in M_t$, $|\tilde{\varepsilon}_i|\le m\rho_1$ and $e_i+\tilde{\varepsilon}_i\ge 0$ to make $\mathbf{f}_t(\mathbf{e}^{\mathbf{G}_{t}}+\boldsymbol{\varepsilon})=\mathbf{e}^{\mathbf{G}_{t}}+\tilde{\boldsymbol{\varepsilon}}$. Next, $\exists \hat{\boldsymbol{\varepsilon}}$ such that for any $i\in M_t$, $|\hat{\varepsilon}_i|\le m^2\rho_1=\rho$ and $e_i+\hat{\varepsilon}_i\ge 0$ to make $\mathbf{f}_t^{(2)}(\mathbf{e}^{\mathbf{G}_{t}}+\boldsymbol{\varepsilon})=\mathbf{e}^{\mathbf{G}_{t}}+\hat{\boldsymbol{\varepsilon}}$. As a result, restrict it to $M_{t+1}$ to get $\mathbf{f}_t^{M_{t+1}(2)}(\mathbf{e}^{\mathbf{G}_{t}}+\boldsymbol{\varepsilon})=\mathbf{e}^{\mathbf{G}_{t+1}}+\hat{\boldsymbol{\varepsilon}}^{\mathbf{G}_{t+1}}$, which means,
\begin{equation*}\label{eq5}
    \mathbf{f}_t^{(2)}(\mathbf{e}^{\mathbf{G}_{t}}+ \boldsymbol{\varepsilon}) = (\mathbf{e}^{\mathbf{G}_{t+1}}+\hat{\boldsymbol{\varepsilon}}^{\mathbf{G}_{t+1}},\mathbf{e}^{*I},\mathbf{0}^{I'})
\end{equation*}

Noticing that after two iterates, the players in $M_{t+1}$ can only have positive-effort neighbors from players in $M_{t+1}$, and thus, 
\[
\mathbf{f}_t^{(3)}(\mathbf{e}^{\mathbf{G}_{t}}+ \boldsymbol{\varepsilon}) = (\mathbf{f}_t^{M_{t+1}}(\mathbf{e}^{\mathbf{G}_{t+1}}+\hat{\boldsymbol{\varepsilon}}^{\mathbf{G}_{t+1}}),\mathbf{e}^{*I},\mathbf{0}^{I'})=(\mathbf{f}_{t+1}(\mathbf{e}^{\mathbf{G}_{t+1}}+\hat{\boldsymbol{\varepsilon}}^{\mathbf{G}_{t+1}}),\mathbf{e}^{*I},\mathbf{0}^{I'}).
\]
Repeat it to know that, 
\begin{equation*}\label{eq6}
    \mathbf{f}_t^{(2+r)}(\mathbf{e}^{\mathbf{G}_{t}}+ \boldsymbol{\varepsilon}) = (\mathbf{f}_{t+1}^{(r)}(\mathbf{e}^{\mathbf{G}_{t+1}}+\hat{\boldsymbol{\varepsilon}}^{\mathbf{G}_{t+1}}),\mathbf{e}^{*I},\mathbf{0}^{I'})\to (\mathbf{e}^{\mathbf{G}_{t+1}},\mathbf{e}^{*I},\mathbf{0}^{I'})=\mathbf{e}^{\mathbf{G}_{t}},
\end{equation*}
which means $\mathbf{e}^{\mathbf{G}_{t}}$ is a stable specialized equilibrium of $\Gamma(\mathbf{G}_{t})$.

On the other hand, if $\mathbf{e}^{\mathbf{G}_{t}}$ is stable, there exists $\rho>0$ such that $\mathbf{f}_{t}^{(n)}(\mathbf{e}^{\mathbf{G}_{t}}+ \boldsymbol{\hat{\varepsilon}})\to \mathbf{e}^{\mathbf{G}_{t}}$ if for any $i\in M_t$, $|\hat{\varepsilon}_i|\le \rho$ and $e_i+\hat{\varepsilon}_i\ge 0$. Let $\rho_1=\rho$, then for any $\boldsymbol{\varepsilon}\in \mathbb{R}^{|M_{t+1}|}$ such that $|\varepsilon_i|\le \rho_1$ and $e_i+\varepsilon_i\ge 0$, denote $\hat{\boldsymbol{\varepsilon}}=(\boldsymbol{\varepsilon},\mathbf{0}^{I\cup I'})$ to know that,
\[
\mathbf{f}_t (\mathbf{e}^{\mathbf{G}_{t}}+ \hat{\boldsymbol{\varepsilon}}) = \mathbf{f}_t (\mathbf{e}^{\mathbf{G}_{t+1}}+ \boldsymbol{\varepsilon},\mathbf{e}^{*I},\mathbf{0}^{I'})
=(\mathbf{f}_{t+1} (\mathbf{e}^{\mathbf{G}_{t+1}}+ \boldsymbol{\varepsilon}),\mathbf{e}^{*I},\mathbf{0}^{I'}).
\]
Repeat it to get,
\[
\mathbf{f}_t^{(r)} (\mathbf{e}^{\mathbf{G}_{t}}+ \hat{\boldsymbol{\varepsilon}})
=(\mathbf{f}_{t+1}^{(r)} (\mathbf{e}^{\mathbf{G}_{t+1}}+ \boldsymbol{\varepsilon}),\mathbf{e}^{*I},\mathbf{0}^{I'})\to (\mathbf{e}^{\mathbf{G}_{t+1}},\mathbf{e}^{*I},\mathbf{0}^{I'}),
\]
which means $\mathbf{e}^{\mathbf{G}_{t+1}}$ is a stable specialized equilibrium of $\Gamma(\mathbf{G}_{t+1})$.

\textbf{Case 2: Only nodes $I''$ is eliminated in step 4 in the $(t+1)$th iteration.} Suppose $M_t=M_{t+1}\cup I''$ where $M_{t+1}$ is the residual players in $\Gamma(\mathbf{G}_{t+1})$. Then,
\begin{equation*}\label{eq7}
    \mathbf{f}_t(\mathbf{e}^{\mathbf{G}_{t}}+ \boldsymbol{\varepsilon}) = (\mathbf{f}_t^{M_{t+1}}(\mathbf{e}^{\mathbf{G}_{t}}+ \boldsymbol{\varepsilon}),\mathbf{f}_t^{I''}(\mathbf{e}^{\mathbf{G}_{t}}+ \boldsymbol{\varepsilon}))
\end{equation*}

By the definition of set $I''$, the players in $M_{t+1}$ can only be influenced by other players in $M_{t+1}$, and the players in $I''$ can also only be influenced by players in $M_{t+1}$, which means $\mathbf{f}_t^{M_{t+1}}(\mathbf{e}^{\mathbf{G}_{t}}+ \boldsymbol{\varepsilon})=\mathbf{f}_{t+1}(\mathbf{e}^{\mathbf{G}_{t+1}}+ \boldsymbol{\varepsilon})$ and $\mathbf{f}_t^{I''}(\mathbf{e}^{\mathbf{G}_{t}}+ \boldsymbol{\varepsilon})=\boldsymbol{g}(\mathbf{e}^{\mathbf{G}_{t+1}}+ \boldsymbol{\varepsilon})$ for some $\boldsymbol{g}: \mathbb{R}^{|M_{t+1}|}\to \mathbb{R}^{|I''|}$. As a result, 
\begin{equation}\label{eq4}
    \mathbf{f}_t^{(r)}(\mathbf{e}^{\mathbf{G}_{t}}+ \boldsymbol{\varepsilon}) = (\mathbf{f}_{t+1}^{(r)}(\mathbf{e}^{\mathbf{G}_{t+1}}+ \boldsymbol{\varepsilon}),\boldsymbol{g}(\mathbf{f}_{t+1}^{(r-1)}(\mathbf{e}^{\mathbf{G}_{t+1}}+ \boldsymbol{\varepsilon}))).
\end{equation}

On the one hand, if $\mathbf{e}^{\mathbf{G}_{t+1}}$ is a stable equilibrium of $\Gamma(\mathbf{G}_{t+1})$, for any $\eta>0$, define $\eta_1=\eta/m>0$, $\exists R_1$ s.t. $\forall r>R_1$, $|\mathbf{f}_{t+1}^{(r-1)}(\mathbf{e}^{\mathbf{G}_{t+1}}+ \boldsymbol{\varepsilon})-\mathbf{e}^{\mathbf{G}_{t+1}}|<\eta_1=\eta/m$ where a vector $\boldsymbol{v}$ and a constant $a$ satisfies $|\boldsymbol{v}|<a$ means that all items $v_i$ in vector $\boldsymbol{v}$ satisfies $|v_i|<a$. Then, $\forall r>R_1$, $|\boldsymbol{g}(\mathbf{f}_{t+1}^{(r-1)}(\mathbf{e}^{\mathbf{G}_{t+1}}+ \boldsymbol{\varepsilon}))-\mathbf{e}^{\mathbf{G}_{t},I}|<\eta$ where $\mathbf{e}^{\mathbf{G}_{t},I}$ is the effort profile of players in set $I$ according to $\mathbf{e}^{\mathbf{G}_{t}}$. Combining those to get that for any $\eta>0$, $\exists R_1$ s.t. $\forall r>R_1$, $|\mathbf{f}_t^{(r)}(\mathbf{e}^{\mathbf{G}_{t}}+ \boldsymbol{\varepsilon})-\mathbf{e}^{\mathbf{G}_{t}}|<\eta$, which means $\mathbf{e}^{\mathbf{G}_{t}}$ is a stable specialized equilibrium of $\Gamma(\mathbf{G}_{t})$.

One the other hand, if $\mathbf{e}^{\mathbf{G}_{t}}$ is a stable specialized equilibrium of $\Gamma(\mathbf{G}_{t})$, then for any $\eta>0$ there exists $R_1$ such that $\forall r>R_1$, $|\mathbf{f}_{t}^{(r)}(\mathbf{e}^{\mathbf{G}_{t}}+ \boldsymbol{\varepsilon})-\mathbf{e}^{\mathbf{G}_{t}}|<\eta$, it follows that $|\mathbf{f}_{t+1}^{(r)}(\mathbf{e}^{\mathbf{G}_{t+1}}+ \boldsymbol{\varepsilon})-\mathbf{e}^{\mathbf{G}_{t+1}}|<\eta$ according to \eqref{eq4}, which proves that $\mathbf{e}^{\mathbf{G}_{t+1}}$ is a stable equilibrium of $\Gamma(\mathbf{G}_{t+1})$.
\end{proof}

\end{document}